%% file: arxiv_main.tex
\documentclass[format=acmsmall, review=false]{acmart}
\usepackage{arxiv_version}
\usepackage{booktabs} 
\usepackage[ruled]{algorithm2e} 

\SetAlFnt{\small}
\SetAlCapFnt{\small}
\SetAlCapNameFnt{\small}
\SetAlCapHSkip{0pt}
\IncMargin{-\parindent}
\acmYear{2026}

\setcitestyle{authoryear}

\title[Fair Data-Exchange Mechanisms]{Fair Data-Exchange Mechanisms}

\author{Rashida Hakim}
\author{Christos Papadimitriou}
\author{Mihalis Yannakakis}

\authorsaddresses{%
  Authors' contact information: rashida.hakim@columbia.edu, christos@columbia.edu, mihalis@cs.columbia.edu \\
  Affiliation: Columbia University, New York, NY, USA}

\renewcommand{\shortauthors}{Hakim, Papadimitriou, and Yannakakis}

\begin{abstract}
We study data exchange among strategic agents without monetary transfers, motivated by domains such as research consortia and healthcare collaborations where payments are infeasible or restricted. The central challenge is to reap the benefits of data-sharing while preventing free-riding that would otherwise lead agents to under invest in data collection. We introduce a simple fair-exchange contract in which, for every pair of agents, each agent receives exactly as many data points as it provides, equal to the minimum of their two collection levels. We show that the game induced by this contract is supermodular under a transformation of the strategy space. This results in a clean structure: pure Nash equilibria exist, they form a lattice, and can be computed in time quadratic in the number of agents. In addition, the maximal equilibrium is truthfully implementable under natural enforcement assumptions and is globally Pareto-optimal across all strategy profiles. In a graph-restricted variant of the model supermodularity fails, but an adaptation of the construction still yields efficiently computable pure Nash equilibria and Pareto-optimal outcomes. Overall, fair exchange provides a tractable and incentive-aligned mechanism for data exchange in the absence of payments.
\end{abstract}

\keywords{data exchange, mechanism design, supermodular games, federated learning}

\begin{document}

\begin{titlepage}

\maketitle


\end{titlepage}

\section{Introduction}
In many domains, valuable data cannot simply be bought and sold. In healthcare, hospitals and research consortia may face legal restrictions on selling patient-level records. In academia, ethical norms and funding rules often prohibit researchers from treating data as a commercial asset. Yet the value of collaboration is clear: larger and more diverse datasets lead to better models, more reliable scientific conclusions, better-trained systems, and greater societal benefit. 

Collecting and cleaning data is often costly, as it requires time, computational resources, and domain expertise. Without monetary payments, data sharing can lead to freeloading: organizations may avoid the upfront costs of collecting and cleaning data, hoping to reap the benefits of others’ work. In this manner, collaborative efforts may even reduce total data collection and social welfare.   

When each party privately knows its own costs and expected benefits, the challenge is to design contracts that coordinate efficient sharing and align incentives for honest participation. Naïve contracts can fail dramatically. Full sharing encourages freeloading, where one party contributes little but benefits from others’ efforts, while overly restrictive sharing can leave mutual gains unrealized. 

We propose a simple and principled alternative: a {\em fair data-exchange mechanism} that enables each pair to share data one-for-one to the maximum extent possible, capped by the smaller contribution. In the resulting continuous\footnote{In view of the large amounts of data involved in the intended applications of this work, it makes sense to assume that data amounts are continuous; however, we also discuss the discrete case.} 
game, a transformation of strategy space turns out to be advantageous: instead of considering the amount of data each player collects, we calculate the total amount that they receive after the exchange is complete. 

We show that this game is supermodular with positive spillovers \cite{levin2003supermodular}, ensuring the existence of a lattice of pure Nash equilibria. We characterize the maximal equilibrium of this game, and we prove that the maximal equilibrium is Pareto-optimal over all outcomes. For completeness, we also characterize the minimal equilibrium. In addition, a mechanism based on the maximal equilibrium of fair-data-exchange is truthful, in that agents have no incentives to misrepresent their private information about costs or their benefit from data. Finally, this analysis can be extended to discrete and network-constrained settings. The fair-exchange mechanism thus offers a tractable foundation for equitable, incentive-compatible data sharing when monetary transfers are problematic.

\subsection{Overview of Results}
Our main contributions are these:
\begin{itemize}
    \item \textbf{Fair-exchange mechanism (Section~\ref{sec:fair-exchange}).} We formalize a simple contract for data-for-data exchange, where each pair of agents shares the minimum of their collected data quantities. This contract ensures reciprocal sharing without transfers and serves as the foundation for our analysis. 

    \item\textbf{Supermodularity (Section~\ref{sec:supermodular-fair-exchange}).} We show that when data amounts are continuous, the resulting game is supermodular with positive spillovers. This structure guarantees the existence of a lattice of pure Nash equilibria and that there exists a Pareto-best NE. 
    
    \item\textbf{Constructive algorithm. (Section~\ref{subsec:max-eq}).} We give a quadratic-time algorithm that computes the maximal equilibrium of the fair-exchange game. This coincides with the Pareto-best NE promised by supermodularity.
    
    \item\textbf{Truthfulness and Pareto Optimality (Section~\ref{sec:truthfulness}).}
    We show that a mechanism which elicits agents’ self-reported data costs and benefit functions and outputs this equilibrium is truthful under natural enforcement models. We also prove that the maximal equilibrium is Pareto-optimal across all strategy profiles.  
    
    \item\textbf{Graph and Discrete Extensions (Sections~\ref{sec:graph} and ~\ref{sec:discrete}).}
    We extend our analysis to graph-restricted and discrete settings. In both cases, supermodularity is lost and equilibria may become incomparable, but modified algorithms still lead to truthful mechanisms and Pareto-undominated outcomes.
\end{itemize}

\subsection{Related Work}
\label{sec:related}

Our work connects most directly to data-exchange economies, including \citet{bhaskara2024data}, \citet{akrami2025theoretical} and \cite{song2025existence}, which study how agents voluntarily exchange data for mutual benefit under fairness (reciprocity) and stability requirements. A key distinction is that these models abstract away the costs of acquiring data, so any restriction on sharing to meet fairness goals is imposed normatively. In contrast, we explicitly model costly data collection, so reciprocal exchange is necessary in order to create good welfare and incentivize honest reporting in the mechanism design problem. 

A complementary line of work studies incentives in federated learning. Recent papers such as \citet{murhekar2024you,yi2025incentive_fl, karimireddy2022mechanisms} and the related models in \citet{chen2023mechanism, clinton2024data} analyze effort, incentives, and welfare when agents jointly train a shared global model. These settings face similar incentive issues as ours, such as free-riding that diminishes social welfare and leads to reduced data contribution from everyone. Federated learning is a public-learning paradigm whereas data exchange is a private-learning paradigm, where each agent has its own task that it cares about. This leads to fundamentally different strategic structures and methods for how to incentivize effort (typically through payments or distributing a corrupted data set). Due to the difference in our model, we are able to rely on the reciprocity of the exchange alone to align incentives for the agents.

We also contribute to the literature on mechanism design for data acquisition with strategic agents. Prior work studies payments for strategic data sources \citep{cai2015opt, ananthakrishnan2024delegating, gong2018incentivizing} and incentive compatible privacy-aware data contribution \citep{cummings2023optimal, fallah2024optimal}. Similar to our setting, agents experience costs for contributing more (or better) data to others. However, these models typically operate in the principal-agent setting, where the person benefiting from the data is different from the person paying the cost to acquire (or share, in the privacy literature) the data. The data exchange setting models the case in which all agents are participating on both sides of this contract. 

Technically, our continuous model draws on the theory of supermodular games and strategic complementarities \citep{topkis1979equilibrium, milgromroberts90}, though the particular structure we induce yields additional properties—global Pareto-optimality and truthful implementation—that do not follow from supermodularity alone. Another paper \cite{alaei2025incentivizing} also finds a connection between a data sharing platform with costly effort and supermodularity, although their method hinges on a payment scheme. This suggests that inducing supermodularity could be a useful design principle when creating incentive-compatible data-exchange rules. 

\section{Model}
\label{sec:model}
There are $n$ agents. Each agent $i$ can privately \emph{collect} data samples from a common underlying distribution, paying a constant per-sample cost $c_i>0$.
We model collection as \emph{continuous}: agent $i$ chooses an effort level $x_i\ge 0$, interpreted as the amount of data collected. We assume a sufficiently large sample pool so that agents' samples do not overlap (or overlap with negligible probability).

Once data are collected, they can be \emph{shared} digitally at zero marginal cost with any number of recipients. Therefore the key friction is not transfer, but the incentive to incur collection costs in the first place. 

\paragraph{Benefits from data.}
Agent $i$'s value from having access to $t_i$ total samples---its own plus any it receives from others---is given by a benefit function $b_i:\mathbb{R}_{\ge 0}\to\mathbb{R}$.
We assume $b_i$ is continuous, nondecreasing, and concave, capturing diminishing marginal returns from repeatedly sampling the same distribution. We further assume the marginal value of additional data vanishes at large scale, formalized as $\lim_{t\to\infty} b_i'(t)=0$. Note that throughout the paper we will interpret $b_i'(t)$ as the right derivative when $b_i$ is not differentiable. 

\subsection{Outcomes and payoffs}
Fix a data-collection profile $X=(x_1,\dots,x_n)$.
A data-sharing contract specifies, for each ordered pair $(i,j)$, how many of $j$'s samples become accessible to $i$ as a function of their collections; we write this as
\[
D_{ij}(x_i,x_j)\;\ge\;0.
\]
Given $X$, agent $i$'s total accessible data is
\[
t_i(X)\;=\;x_i \;+\; \sum_{j\neq i} D_{ij}(x_i,x_j),
\]
and her utility is quasi-linear in collection costs:
\[
U_i(X)\;=\; b_i\!\big(t_i(X)\big)\;-\;c_i x_i.
\]

\paragraph{Why not always fully share?}
Because $b_i$ is nondecreasing and replication is free, \emph{conditional on a fixed realized dataset}, it is socially optimal to make all collected samples accessible to everyone.
However, unconditional sharing lets agents enjoy others' data without paying collection costs, creating a free-riding problem at the \emph{collection} stage. Examples of catastrophic freeloading of this type are given by \cite{chen2023mechanism} and \cite{karimireddy2022mechanisms}. Agreeing in advance to exchange contracts aims to recover some of the efficiency of broad access while maintaining incentives to contribute data.



\subsection{Fair Exchange Contract}
\label{sec:fair-exchange}
We now introduce the central idea studied in this paper: the \emph{fair-exchange contract}, which we view as an interpolation between no-exchange and full-exchange which induces desirable properties. Under this contract, each pair of agents exchanges as much data as both have collected, capturing the idea of reciprocal contribution. Formally, for any pair of agents $i$ and $j$, the contract specifies that the number of samples transferred from $j$ to $i$ is
\[
D_{ij}(x_i, x_j) = \min\{x_i,\, x_j\}.
\]

Under fair exchange, each agent’s total accessible data is
\[
t_i(x)
= x_i + \sum_{j \neq i} \min\{x_i, x_j\}.
\]

Because data transfer is costless, an agent’s strategic decision reduces to choosing $x_i$, the amount of data it personally collects. The fair-exchange rule determines how this interacts with others’ choices.

\subsubsection{Two Perspectives on the Fair-Exchange Game}
In \emph{data-collection space}, the fair-exchange contract defines a normal-form game in which each agent $i$ chooses a collection level $x_i$.  Given a profile $X=(x_1,\dots,x_n)$, the contract deterministically maps collections to an exchange outcome and hence to payoffs $u_i(X)$ (value from accessed data minus collection cost).

In our analysis, it is often convenient to work in \emph{total-access coordinates}.  Given $X$, let $T=\Phi(X)=(t_1,\dots,t_n)$ denote the induced vector of total accessible data levels.  Appendix~\ref{fwdinv} provides an inverse map $\Phi^{-1}$ on the feasible image of $\Phi$, so each $T$ corresponds to a unique collection profile $X=\Phi^{-1}(T)$ and hence to the same exchange outcome and payoffs. Agents' strategic choices remain the collection levels $x_i$: the map $X \leftrightarrow T=\Phi(X)$ is a profile-level reparameterization used for analysis.

\subsubsection{Agent Ranking}
\label{def:rankings}
We track agents' incentives using only their index in the sorted data collection profile $X$. For a profile $X=(x_1,\dots,x_n)$, we define two useful quantities that describe agents' exact positions, since a simple index isn't sufficient to account for ties. 
\[
k_i(X)\;:=\;\bigl|\{\,j:\; x_j \ge x_i\,\}\bigr|\in\{1,\dots,n\}
\qquad\text{(\# agents weakly above, including self),}
\]
\[
k_i^{\uparrow}(X)\;:=\;1+\bigl|\{\,j:\; x_j > x_i\,\}\bigr|\in\{1,\dots,n\}
\qquad\text{(\# agents strictly above, plus self).}
\]

The weakly-above count governs incentives to decrease contributions, as an agent loses out on that much data per unit that they marginally decrease due to fair-exchange. The strict-above count likewise governs incentives to \emph{increase} contributions, since agents can gain that much for every additional marginal unit they collect. 

We frequently work in total-data space $T=\Phi(X)$. Importantly, $\Phi$ preserves the ranking of agents: agents who collect more get the same amount from smaller agents and more from bigger agents. Therefore, if $x_i<x_j$ then $t_i<t_j$ (and due to the symmetry of fair-exchange, ties in $x$ correspond to ties in $t$). Thus the quantities $k_i(\cdot)$ and $k_i^{\uparrow}(\cdot)$ can be equivalently defined using $T$ instead of $X$.

Next, we will show the game is supermodular when analyzed from the total-data perspective, implying the existence of a lattice of pure Nash equilibria. 

\section{Supermodularity}
\label{sec:supermodularity}

\subsection{Definitions}

\begin{definition}[Supermodularity]
A game is supermodular (equivalently has increasing differences) if raising other agents’ actions makes agent $i$’s marginal gain from raising her own action weakly larger \cite{topkis1979equilibrium}. 

Formally, if for every agent $i$, for all $X_{-i} \le X'_{-i}$ and all $x_i \le x'_i$,
\[
U_i(x'_i,\, X_{-i}) - U_i(x_i,\, X_{-i})
\;\le\;
U_i(x'_i,\, X'_{-i}) - U_i(x_i,\, X'_{-i}) .
\]

A key implication is that best responses must be \emph{increasing}: if other agents choose larger strategies, $i$'s best response cannot fall. 
\end{definition}

\begin{definition}[Positive spillovers]
A game has \emph{positive spillovers} if increasing one agent's strategy weakly increases all others' payoffs. 
Equivalently, for any distinct agents $i\neq j$, for all $x_j \le y_j$, and for all fixed $x_{-j}\in X_{-j}$,
\[
U_i(x_j,\, x_{-j}) \;\le\; U_i(y_j,\, x_{-j}) .
\]
\end{definition}

\subsection{Application to Fair Exchange}
\label{sec:supermodular-fair-exchange}
\paragraph{Supermodularity fails in data-collection space.}
The fair-exchange game violates this property when strategies are expressed in $(x_i)$. For example, consider two agents where agent~1 wants to end with a total of $10$ data points. If agent~2 collects $5$, agent~1 best responds with $5$; if agent~2 collects $0$, agent~1 best responds with $10$. Thus agent~1's best response \emph{decreases} when agent~2 collects more, contradicting increasing best responses. 

\paragraph{Supermodularity in total-data space.} This failure of supermodularity turns out to be product of the data-collection lens. We can observe in our example above that no agent receives less total data in the second case, even though agent 1's best response decreases. This observation leads to a more general theorem.

\begin{theorem}[Supermodularity in total-access coordinates]\label{thm:supermod}
Under the partial order on collection profiles defined by
$X \preceq X' \iff \Phi(X) \le \Phi(X')$ componentwise,
the fair-exchange game is supermodular.
\end{theorem}

\begin{proof}
We will abuse notation by writing
$x_i(T):=\Phi^{-1}(T)_i$, where $\Phi^{-1}$ is the inverse transform from total-data space to data collection space (defined in Appendix~\ref{fwdinv}).

Fix an agent $i$. Since $U_i(T)=b_i(T_i)-c_i\,x_i(T)$ and $b_i(T_i)$ depends only on $T_i$,
it suffices to show that $x_i(T)$ has \emph{decreasing differences} in $(T_i,T_{-i})$:
for all $t_i\le t_i'$ and all $T_{-i}\le T'_{-i}$,
\[
x_i(t_i',T'_{-i})-x_i(t_i,T'_{-i})
\;\le\;
x_i(t_i',T_{-i})-x_i(t_i,T_{-i}).
\tag{DD}\label{eq:dd}
\]

\textbf{Step 1: Affine structure of $x_i$ as $T_i$ varies.}
Fix $T_{-i}$ and increase $t_i$ continuously. As long as $t_i$ stays between two adjacent coordinates of $T_{-i}$,
the set of agents with total at least $t_i$ is unchanged; equivalently, the quantity $k_i(t_i,T_{-i})$ is constant on that interval.

On such an interval, the behavior of agents strictly below $i$ is pinned down by their own totals, so the amount of data
they contribute upward is constant as $t_i$ varies; let $P$ denote this constant ``baseline'' contribution to $t_i$ coming
from agents below $i$. Any additional total for $i$ beyond $P$ must be generated by $i$'s own collection.
Under fair exchange, each marginal unit collected by $i$ is delivered to exactly the agents weakly above $i$,
i.e., to $k_i(t_i,T_{-i})$ recipients. Therefore, within this interval, increasing $t_i$ by one unit requires
increasing $x_i$ by exactly $1/k_i(t_i,T_{-i})$. Formally, we can derive this from the structure of $\Phi^{-1}$ given in Appendix~\ref{fwdinv}.
\[
x_i(t_i,T_{-i})=\frac{t_i-P}{k_i(t_i,T_{-i})}
\quad\text{on this interval,}
\]
so $x_i(\cdot,T_{-i})$ is affine there with slope $1/k_i(t_i,T_{-i})$.

\textbf{Step 2: Increasing $T_{-i}$ makes the slope smaller.}
Now take $T'_{-i}\ge T_{-i}$ and fix $t_i$. If other agents' totals increase, then (weakly) more agents satisfy $T_j\ge t_i$,
so
\[
k_i(t_i,T'_{-i}) \;\ge\; k_i(t_i,T_{-i}).
\]
Since the local slope equals $1/k_i(\cdot)$, it follows that wherever $k_i$ is locally constant,
\[
\frac{\partial}{\partial t_i}\,x_i(t_i,T'_{-i})
\;\le\;
\frac{\partial}{\partial t_i}\,x_i(t_i,T_{-i}).
\]

\textbf{Step 3: Decreasing differences.}
Integrating this slope inequality over $t_i\in[t_i,t_i']$ (noting that $x_i(\cdot,T_{-i})$ and $x_i(\cdot,T'_{-i})$
are piecewise affine with finitely many points where $k_i(\cdot)$ changes) yields
\[
x_i(t_i',T'_{-i})-x_i(t_i,T'_{-i})
\;\le\;
x_i(t_i',T_{-i})-x_i(t_i,T_{-i}),
\]
which is \eqref{eq:dd}. Multiplying by $-c_i$ gives increasing differences for $-c_i x_i(T)$, and therefore
$U_i$ has increasing differences in $(T_i,T_{-i})$.
\end{proof}

\paragraph{Positive spillovers in total-data space}While it is immediate that the game has positive spillovers in data-collection space, it requires proof in total-data space due to the effects of what other agents report on how much an agent is asked to collect.

\begin{theorem}
When strategies are expressed as total data $T=(t_1,\dots,t_n)$, the fair-exchange game has positive spillovers.
\end{theorem}

\begin{proof}
Fix an agent $i$. We must show that for all $T_{-i}\le T'_{-i}$ and all fixed $t_i$,
\[
U_i(t_i,T'_{-i}) \;\ge\; U_i(t_i,T_{-i}).
\tag{PS}\label{eq:ps}
\]
Since $U_i(T)=b_i(t_i)-c_i\,x_i(T)$ and $b_i(t_i)$ does not depend on $T_{-i}$, it suffices to prove that
for fixed $t_i$,
\[
x_i(t_i,T'_{-i}) \;\le\; x_i(t_i,T_{-i}).
\tag{PS$'$}\label{eq:psprime}
\]

Fix $t_i$. On any region of $T_{-i}$ where the set of agents with total at least $t_i$ is unchanged  (so $k_i(t_i,T_{-i})$ is constant), the inverse formula for $\Phi^{-1}$ implies that
\[
x_i(t_i,T_{-i})=\frac{t_i-P}{k_i(t_i,T_{-i})},
\]
where $P$ is the baseline contribution to $t_i$ coming from agents below $i$. Now increase $T_{-i}$ to $T'_{-i}\ge T_{-i}$.
Then the baseline contribution from below can only increase ($P'\ge P$), and the set of agents with total at least $t_i$
can only grow (so $k_i(t_i,T'_{-i})\ge k_i(t_i,T_{-i})$). Both effects weakly decrease
$\frac{t_i-P}{k_i}$, hence $x_i(t_i,T'_{-i})\le x_i(t_i,T_{-i})$, proving \eqref{eq:psprime} and therefore \eqref{eq:ps}.
\end{proof}

\subsection{Consequences for Equilibria}

We rely on the following consequences of supermodularity (all in total-data space):

\begin{enumerate}
    \item \textbf{Pure Nash equilibrium exists} \cite{topkis1979equilibrium}.  
    \item \textbf{Equilibria form a lattice}: there is a well-defined smallest and largest equilibrium under componentwise order \cite{topkis1979equilibrium}.  
    \item \textbf{With positive spillovers, the largest equilibrium is Pareto-best}: meaning \emph{every} agent
    weakly prefers it to every other equilibrium \cite{vives05}. Note that this is a stronger property than Pareto-optimality. 
\end{enumerate}

These properties apply to the maximal equilibrium that we give a constructive algorithm for in Section~\ref{sec:equilibria}. 

\section{Equilibria}
\label{sec:equilibria}
We begin our analysis of equilibria of the game with an important definition that captures an agent's marginal tradeoff at a given rank parameter. 

\begin{definition}[$K$-data level]
\label{def:kdatalevel}
For each agent $i$ and integer $K \ge 0$, the \emph{$K$-data level}, denoted $s_i^K$, is defined by
\[
s_i^K \;:=\; \max\Bigl\{\,s \ge 0 \;:\; b_i'(s)\geq\tfrac{c_i}{K}\Bigr\}.
\]
\end{definition}

The $K$-data-level $s_i^K$ is the amount of total accessible data at which agent $i$ no longer wants to collect a marginal unit when it would receive $K$ \emph{additional} units for free through fair exchange. This is exactly the situation an agent faces when $K$ agents are strictly above them. Because $b_i$ is concave, the sequence $(s_i^K)_{K=0}^{n-1}$ is nondecreasing.

The $K$-data-levels are the main ingredient to understanding the local equilibrium conditions for each agent. They will also prove critical for computing equilibria of the fair-exchange game. 

\subsection{Local equilibrium conditions}
\label{sec:local-cond}
In any pure Nash equilibrium $X$, every agent must satisfy two local incentive constraints. These conditions are clearly necessary but not immediately sufficient to prove that a particular profile is a PNE. However, we will later argue that these two conditions are actually sufficient due to the structure of the fair-exchange game. 
\begin{itemize}
    \item \textbf{No profitable upward deviation.}
    If agent $i$ has $k_i^{\uparrow}(X)$ agents collecting strictly more than it does, then
    \[
    t_i(X^\ast) \ge s_i^{k_i^{\uparrow}(X)}.
    \]
    \item \textbf{No profitable downward deviation.}
    If agent $i$ has $k_i(X)$ agents collecting at least as much as it does, then
    \[
    t_i(X^\ast) \le s_i^{k_i(X)}.
    \]
\end{itemize}

\subsection{Deviation Lemmas}
In order to prove that the local conditions suffice and our later results about truthfulness and global Pareto optimality, we introduce two key deviation lemmas that capture the fundamental structure of the fair-exchange game. 
\begin{lemma}[Upward move is harmful when lower agents are fixed]\label{lem:upward-harm-fixed-lower-plusone}
Fix an agent $i$ and two collection profiles $X$ and $X'$ with $x_i'>x_i$. Assume that for every agent $j\neq i$ with $x_j\le x_i$ we have $x_j'=x_j$ (i.e., all agents other than $i$ who collect weakly less than $x$ in $X$ keep the same collection in $X'$). Additionally assume agent $i$ strictly does not want to locally deviate upwards in $X$, meaning that if they collected marginally more data, holding others fixed, their utility would strictly decrease. These conditions imply agent $i$ is strictly worse off in $X'$ than in $X$. 
\end{lemma}

\begin{proof}
Because every agent $j\neq i$ with $x_j\le x_i$ keeps the same collection in $X'$, their total contribution to $i$'s accessible data is identical across $X$ and $X'$. Thus any increase in $i$'s total accessible data when moving from $X$ to $X'$ can only come from raising $i$'s own collection from $x_i$ to $x'_i$ and from the agents who are above $i$.

In the best case for agent $i$ in $X'$, the agents strictly above $x_i$ in $X$ are all above $x'_i$ in $X'$ (due to positive spillovers). In this case, $i$'s total accessible data will increase by $k_i^{\uparrow}(X)(x'_i - x_i)$ where $k_i^{\uparrow}$ is exactly the number of agents strictly above $i$, plus $i$ themselves (as defined in Section~\ref{def:rankings}). This is because under fair-exchange, the additional amount an agent will contribute to $i$ for increasing their data collection is exactly the amount of the increase, and raising $i$'s level cannot create new agents above $i$.
\[
t_i(X') \;\le\; t_i + \bigl(k_i^{\uparrow}(X)\bigr)(x_i'-x_i).
\]

Hence $i$'s benefit gain is at most
\[
b_i\!\left(t_i+\bigl(k_i^{\uparrow}(X)\bigr)(x_i'-x_i)\right)-b_i(t_i),
\]
while its cost increases by exactly $c_i(x_i'-x_i)$. By concavity of $b_i$,
\[
b_i\!\left(t_i+\bigl(k_i^{\uparrow}(X)\bigr)(x_i'-x_i)\right)-b_i(t_i)
\;\le\;
b_i'(t_i)\cdot \bigl(k_i^{\uparrow}(X)\bigr)(x_i'-x_i).
\]

Combining,
\[
\bigl(b_i(t_i(X'))-c_i x_i'\bigr)-\bigl(b_i(t_i)-c_i x_i\bigr)
\;\le\;
\Bigl(\bigl(k_i^{\uparrow}(X)\bigr)b_i'(t_i)-c_i\Bigr)(x_i'-x_i)
\;<\; 0,
\]
where the strict inequality uses $\bigl(k_i^{\uparrow}(X)\bigr)b_i'(t_i)<c_i$ and $x_i'-x_i>0$.
Thus $i$ is strictly worse off in $X'$ than in $X$.
\end{proof}

\begin{lemma}[Downward move is weakly harmful when low agents are fixed]\label{lem:downward-harm-robust}
Fix an agent $i$ and two collection profiles $X$ and $X'$ with $x_i'<x_i$.
Assume that for every $j\neq i$ with $x_j < x'_i$ we have $x_j(X')=x_j(X)$ (i.e., all agents other than $i$ who collect strictly less than $x_i'$ in $X$ keep the same collection in $X'$). Additionally assume agent $i$ weakly does not want to locally deviate downwards in $X$, meaning that if they collected marginally less data, holding others fixed, their utility would weakly decrease. These conditions imply agent $i$ is weakly worse off in $X'$ than in $X$. 
\end{lemma}

\begin{proof}
Consider profile $X$, but with agent $i$ instead playing $x'_i$.
Under fair exchange, $i$ receives $\min\{x'_i,x_j\}$ from each $j\neq i$.
By assumption, all agents with $x_j< x_i'$ keep the same collection in $X'$, while agents with $x_j\ge x_i'$ may change.
Thus for every $j\neq i$,
\[
\min\{x_i',x_j'\} \;\le\; \min\{x_i',x_j\},
\]
so $i$'s total data (and hence benefit) from collecting $x'_i$ in $X'$ is weakly smaller than in $X$, while the cost is exactly $c_i x_i'$.

Now hold opponents fixed at $X_{-i}$ and increase $i$'s collection from $x_i'$ up to $x_i$.
For any intermediate level $y\in[x'_i,x_i]$, the set $\{j\neq i: x_j\ge y\}$ is a superset of $\{j\neq i: x_j\ge x_i\}$, so each marginal unit of $i$'s collection leads to receiving a unit of data from at least $k_i(X)$ agents (including themselves). Since $b_i$ is concave, $b_i'$ is nonincreasing, and totals along this path are at most $t_i$, the slope of the benefit function in this region is at least $b'(t_i)$. Hence the marginal payoff along $[x_i',x_i]$ is at least $k_ib_i'(t_i)-c_i\ge 0$, implying that $i$'s payoff in $X$ is weakly increasing as $i$ increases from $x'_i$ to $x_i$.

Combining,
\[
U_i(X) \;\ge\; U_i(x'_i, X_{-i}) \;\ge\; U_i(X'),
\]
as claimed.
\end{proof}

\subsection{Local constraints suffice}
Although Nash equilibrium is defined by the absence of \emph{any} unilateral deviation, in the fair-exchange game it is enough to check the two one-sided constraints defined in Section~\ref{sec:local-cond}, and Lemma~\ref{lem:upward-harm-fixed-lower-plusone} and Lemma~\ref{lem:downward-harm-robust} capture the reason why.

Intuitively, once $X_{-i}$ is fixed, agent $i$'s payoff as a function of $x_i$ is single-peaked: moving upward becomes (strictly) unprofitable past the upward threshold, and moving downward is (weakly) unprofitable below the downward threshold.

Formally, Lemma~\ref{lem:upward-harm-fixed-lower-plusone} implies that if the "no profitable upward deviation" inequality holds at $X$, then \emph{no larger upward jump} can be profitable (with opponents fixed), and Lemma~\ref{lem:downward-harm-robust} implies the analogous statement for downward jumps.
Thus the two inequalities from Section~\ref{sec:local-cond} are equivalent to $i$ best-responding to $X_{-i}$.

\subsection{Algorithms}
Because the fair-exchange game is supermodular in total-data space, the set of pure Nash equilibria is nonempty and forms a lattice. In particular, there is a well-defined \emph{maximal} equilibrium and a well-defined \emph{minimal} equilibrium (with respect to the product order on total-data profiles). In this section we give an explicit, polynomial-time construction of the maximal equilibrium, which is the equilibrium used by our mechanism and in our subsequent Pareto-optimality results. For completeness, we also obtain an explicit construction of the minimal equilibrium; we defer this algorithm and proof to Appendix~\ref{app:min}.

Our algorithms operate in total-data space: given the $K$-data-levels $\{s_i^K\}_{i\in[n],\,K\in\{1,\dots,n\}}$ (defined in Section~\ref{sec:model}), they output a total-data equilibrium vector $T=(t_1,\dots,t_n)$. When needed, we convert $T$ to a data-collection profile $X=(x_1,\dots,x_n)$ via the inverse transform $X=\Phi^{-1}(T)$ (Appendix~\ref{fwdinv}).

\subsubsection{Maximal Equilibrium}
\label{subsec:max-eq}

The maximal equilibrium is the greatest fixed point of the best-response operator in total-data space. The intuition behind Algorithm~\ref{alg:max} is that agents are potentially willing to collect more when others do. So we can ask: how high can the group go together? The limiting factor is a bottleneck agent—the one who is willing to go the least high even in the most favorable situation, when they can exchange with everyone. The algorithm identifies this bottleneck and fixes them at that level. With that agent fixed, we ask the same question for the remaining agents: how high can they go together, given that the fixed agents stay below. Repeating this process peels off bottlenecks one by one, building up the maximal equilibrium.

\begin{algorithm}[H]
  \SetAlgoNoLine
  \KwIn{$K$-data-levels $\{s_j^{\,K}\}_{j\in[n],\,K\in\{1,\dots,n\}}$}
  \KwOut{Total-data equilibrium vector $T \in \mathbb{R}^n_{\ge 0}$}

  $R \gets \{1,\dots,n\}$\tcp*[r]{remaining agents}
  $\textit{prev} \gets 0$\;

  \For{$\tau \gets 1$ \KwTo $n$}{
    $K \gets n-\tau+1$\tcp*[r]{rank parameter}
    select $j^{\*} \in R$ minimizing $s_j^{\,K}$\;
    $T_{j^{\*}} \gets \max\{\,s_{j^{\*}}^{\,K},\ \textit{prev}\,\}$\;
    $\textit{prev} \gets T_{j^{\*}}$\;
    $R \gets R \setminus \{j^{\*}\}$\;
  }

  \caption{Maximal Data Equilibrium}
  \label{alg:max}
\end{algorithm}

Let $T^{\max}$ be the output of Algorithm~\ref{alg:max} and let $X^{max} = \Phi^{-1}(T^{max})$. 

\begin{theorem}[Correctness of Algorithm~\ref{alg:max}]
$\Phi(T^{\max}) = X^{max}$ is a Nash equilibrium, there are no profitable deviations for agents in data collection space.
\label{thm:max-alg-ne}
\end{theorem}

\begin{proof}
Rename agents according to their selection order in Algorithm~\ref{alg:max} as $(1),(2),\dots,(n)$, where $(\tau)$ is the agent selected in iteration $\tau$.

We prove by induction on $\tau$ that after iteration $\tau$, every selected agent $(1),\dots,(\tau)$ satisfies the no-upward and no-downward deviation inequalities in the final profile $T^{\max}$.

\emph{Base case ($\tau=1$).}
At $\tau=1$ the algorithm uses $K=n$ and assigns $T^{\max}_{(1)}=\max\{s_{(1)}^{\,n},0\}=s_{(1)}^{\,n}$.
Since $k_{(1)}^{\uparrow}(T^{\max})\le n$, and $s_{(1)}^K$ is nondecreasing in $k$, we have
$s_{(1)}^{\,k_{(1)}^{\uparrow}(T^{\max})}\le s_{(1)}^{\,n}=T^{\max}_{(1)}$, so the no-upward inequality holds.
The no-downward inequality holds as well by the definition of the downward constraint and $s_{(1)}^K$ being nondecreasing in $K$.

\emph{Inductive step.}
Assume agents $(1),\dots,(\tau-1)$ satisfy both inequalities. Consider agent $(\tau)$.
At iteration $\tau$, the algorithm uses $K=n-\tau+1$ and sets
\[
T^{\max}_{(\tau)}\;=\;\max\{\,s_{(\tau)}^{\,K},\ T^{\max}_{(\tau-1)}\,\}.
\]

\emph{No upward deviation.}
By construction $T^{\max}_{(\tau)}\ge s_{(\tau)}^{\,K}$. In the final profile, at most the $n-\tau$ agents selected after $\tau$ can be strictly above $(\tau)$, hence
\[
k_{(\tau)}^{\uparrow}(T^{\max}) \;=\; 1+\big|\{h\neq(\tau):T^{\max}_h>T^{\max}_{(\tau)}\}\big|
\;\le\;1+(n-\tau)\;=\;K.
\]
Since $s_{(\tau)}^K$ is nondecreasing in $K$, we have $s_{(\tau)}^{\,k_{(\tau)}^{\uparrow}(T^{\max})}\le s_{(\tau)}^{\,K}$, and therefore
\[
T^{\max}_{(\tau)}\ge s_{(\tau)}^{\,K}\ge s_{(\tau)}^{\,k_{(\tau)}^{\uparrow}(T^{\max})},
\]
proving the no-upward inequality.

\emph{No downward deviation.} First, we consider the case where $T^{\max}_{(\tau)}=s_{(\tau)}^{\,K}$. Since later agents are assigned weakly higher levels, we have $k_{(\tau)}(T^{\max})\ge K$, and because $s_{(\tau)}^K$ is nondecreasing in $K$, it follows that
\[
T^{\max}_{(\tau)}=s_{(\tau)}^{\,K}\le s_{(\tau)}^{\,k_{(\tau)}(T^{\max})}.
\]

Second we consider the case where  $T^{\max}_{(\tau)}=prev =T^{\max}_{(\tau-1)}$. Then $(\tau-1)$ is weakly above $(\tau)$ and all agents selected after $\tau$ are weakly above as well, so
\[
k_{(\tau)}(T^{\max})\ge (n-\tau)+2 = K+1
\quad\Longrightarrow\quad
s_{(\tau)}^{\,k_{(\tau)}(T^{\max})}\ge s_{(\tau)}^{\,K+1}.
\]
Thus it suffices to show $T^{\max}_{(\tau)}\le s_{(\tau)}^{\,K+1}$. By induction, we know $(\tau-1)$ does not want to deviate downwards. In addition, we know from $(\tau-1)$ being selected instead of $(\tau)$ at iteration $\tau-1$, that $s_{(\tau-1)}^{\,K+1} \leq s_{(\tau)}^{\,K+1}$.
\[
T^{\max}_{(\tau)}=T^{\max}_{(\tau-1)}
\;\le\; s_{(\tau-1)}^{\,k_{(\tau-1)}(T^{\max})}
\;\le\; s_{(\tau-1)}^{\,K+1}
\;\le\; s_{(\tau)}^{\,K+1}.
\]
Combining gives
\[
T^{\max}_{(\tau)} \le s_{(\tau)}^{\,K+1} \le s_{(\tau)}^{\,k_{(\tau)}(T^{\max})},
\]
which is the no-downward inequality for $(\tau)$.
\end{proof}

\begin{theorem}[Extremality of maximal equilibrium]
In $X^{\max}$, each agent has weakly more total accessible data than they do in any other equilibrium. In other words, no other equilibrium $X'$ exists where $T' = \Phi(X')$ is higher than $T^{max}$ in \emph{any} component.
\label{thm:max-alg-extremal}
\end{theorem}

\begin{proof}
Rename agents as we do in the proof of Theorem~\ref{thm:max-alg-ne}. Suppose for contradiction that some equilibrium $T'$ satisfies
$T'_i>T^{\max}_i$ for some agent, and let $\tau$ be the first index with
$T'_{(\tau)}>T^{\max}_{(\tau)}$. For every $r<\tau$,
\[
T'_{(r)}\le T^{\max}_{(r)}\le T^{\max}_{(\tau)}<T'_{(\tau)},
\]
so all earlier agents are strictly below $(\tau)$ in $T'$. Hence
$k_{(\tau)}(T')\le k_{(\tau)}(T^{\max})$, and $K$-data-levels being nondecreasing in $K$ gives
\[
s_{(\tau)}^{\,k_{(\tau)}(T')}
\le
s_{(\tau)}^{\,k_{(\tau)}(T^{\max})}.
\]
In Algorithm~\ref{alg:max}, when $(\tau)$ is selected the rank parameter is
$K=n-\tau+1$, which equals $k_{(\tau)}(T^{\max})$ (the algorithm's assigned values
are nondecreasing), and the assignment rule sets
\[
T^{\max}_{(\tau)}=\max\{\,s_{(\tau)}^{\,K},\ \textit{prev}\,\}\ge s_{(\tau)}^{\,K}
=
s_{(\tau)}^{\,k_{(\tau)}(T^{\max})}.
\]
Therefore,
\[
s_{(\tau)}^{\,k_{(\tau)}(T')}
\le
T^{\max}_{(\tau)}
<
T'_{(\tau)},
\]
violating the no-downward-deviation constraint
$T'_{(\tau)}\le s_{(\tau)}^{\,k_{(\tau)}(T')}$.
This contradiction implies that every equilibrium $T'$ satisfies $T'\le T^{\max}$ coordinatewise.
\end{proof}

\subsubsection{Minimal Equilibrium (overview)}
\label{subsec:min-eq}

By lattice structure, there is also a well-defined minimal equilibrium $T^{\min}$.
We provide an explicit $O(n^2)$ algorithm computing $T^{\min}$, together with a proof of correctness and minimality, in Appendix~\ref{app:min}.
This equilibrium plays a secondary role in our analysis; the maximal equilibrium $T^{\max}$ is the object implemented by our mechanism and used in our Pareto-optimality results.

\subsubsection{Computational Complexity}
\label{subsec:complexity}

Algorithm~\ref{alg:max} fixes one agent per iteration, with each iteration performing a selection over the remaining agents and simple updates; this yields $O(n^2)$ time.
The minimal equilibrium can be computed in the same asymptotic time (Appendix~\ref{alg:min}).
More generally, while supermodularity guarantees the existence of extremal equilibria and supports monotone-iteration methods, such generic procedures can require exponentially many iterations in the worst case for succinctly represented games \cite{etessami2020tarski}. In contrast, the fair-exchange game has additional structure that yields the explicit $O(n^2)$ construction above.

\subsection{Connection to supermodularity and subsequent results}
\label{subsec:equilibria-connection}
Algorithm \ref{alg:max} gives an explicit construction for the maximal pure NE, whose existence is guaranteed by the lattice structure promised by supermodularity. Under positive spillovers the maximal equilibrium is Pareto-best among equilibria. The maximal equilibrium has surprisingly good properties even beyond what supermodularity promises, which we will see in our mechanism design and global Pareto-optimality analysis.

\section{Truthfulness and Pareto Optimality}
\label{sec:truthfulness}
We prove desirable results for $X^{\max}$ beyond what is guaranteed by supermodularity alone by taking a mechanism design lens as well as a broader welfare-analysis lens. 

We mention here two important properties of the output of Algorithm~\ref{alg:max} that we use in our analysis. These are fairly easy to see in the continuous case, and we re-use them for the graph and discrete extensions. 

The first is \emph{prefix-independence}, which means that two runs of Algorithm~\ref{alg:max} that select identical agents until iteration $\tau$, will place the first $\tau - 1$ agents selected at the same data collection levels. This is true by the bottom up construction of Algorithm~\ref{alg:max} together with the bottom up construction of Algorithm~\ref{alg:inverse} which computes $\Phi^{-1}$.

The second property is \emph{iteration-highest-response}. This means that the agent $i$ selected in iteration $\tau$ of the algorithm strictly does not want to deviate upwards from $x_i^{\max}$ in any profile where all agents selected before $\tau$ are fixed to their data collection levels and all remaining agents are weakly above. For Algorithm~\ref{alg:max}, we can see this property from the fact that, at iteration $\tau$, the number of remaining agents is exactly $K$ and we assign $t_i^{\max}=\max\{s_i^{K},\textit{prev}\} \geq s_i^{K}$ and the definition of $K$-data-levels.

\subsection{Mechanism Models}
We interpret the maximal-equilibrium construction as a direct mechanism. We assume that each agent knows their own cost and benefit function and can therefore privately compute their vector of $K$-data-levels (as in Definition~\ref{def:kdatalevel}). This is a sufficient summary of the benefit function for algorithmic purposes. 

Each agent reports $\hat s_i=(\hat s_i^0,\dots,\hat s_i^{n-1})$, and the mechanism computes the maximal total-data equilibrium $T^{\max}(\hat s)$ via Algorithm~\ref{alg:max}. It then defines the recommended (and in Model~1, enforced) collection profile
\[
X^{\max}(\hat s) \;:=\; \Phi^{-1}\!\big(T^{\max}(\hat s)\big).
\]
In Models~2--3, after observing $X^{\max}(\hat s)$, agents choose how much data to \emph{submit} to the exchange stage; we denote the submitted profile by $X=(x_1,\dots,x_n)$.

\begin{itemize}
  \item \textbf{Model 1 (full enforcement).}
  The mechanism enforces $X = X^{\max}(\hat s)$, i.e., each agent must submit exactly the recommended number of samples.
  \item \textbf{Model 2 (partial enforcement).}
  The recommendation $x_i^{\max}(\hat s)$ serves as both a \emph{participation threshold} and a \emph{cap}.
  If $x_i < x_i^{\max}(\hat s)$ then agent $i$ receives no data from others.
  If $x_i \ge x_i^{\max}(\hat s)$ then $i$ participates, but for the purposes of fair exchange $i$ is treated as contributing exactly $x_i^{\max}(\hat s)$ (so submitting more does not increase $i$'s exchange credit).
  This captures settings where the mechanism cannot prevent off-platform data collection, but can restrict centralized exchange to agents who meet the recommended contribution level.

  \item \textbf{Model 3 (no enforcement).}
  The mechanism only recommends $X^{\max}(\hat s)$. Agents may submit any $X$, and fair exchange is applied to the submitted data.
\end{itemize}
\paragraph{Why doesn't full enforcement imply implementability?} Model 1 sounds extremely strong because the designer can compel agents to take prescribed actions. But this power is only useful if the designer knows which action profile \emph{should} be prescribed, and in our setting the relevant costs and benefits are private information. Thus the key challenge is truthful elicitation: the mechanism must choose an enforced outcome as a function of reported preferences, and agents may have incentives to misreport whenever reports affect who bears costly collection effort versus who gets to benefit from others’ effort. For example, if the designer tries to implement the equilibria of the naturally defined "full-exchange" contract, this will not lead to a truthful mechanism.

\subsection{Truthfulness Under Enforcement}

We show that Models~1 and~2 guarantee truthful reporting of $K$-data-levels. The key idea is that the maximal equilibrium gives each agent the highest total-data level that can be sustained while still being incentive compatible. While an agent $i$ might hope to gain by inducing higher contributions from agents below them, doing so necessarily pushes $i$ lower than those agents in the construction. This eliminates the ability to turn those agents' extra collection into additional benefits.

\begin{theorem}[Truthfulness under enforcement]
\label{thm:truthful}
Under Models~1 and~2, truthful reporting is a Nash equilibrium of the reporting game: for all $i$, all $S=(s_i,s_{-i})$, and all misreports $\hat s_i$,
\[
U_i(S) \;\ge\; U_i(\hat s_i, s_{-i}),
\]
where $U_i(\hat s_i, s_{-i})$ denotes $i$'s \emph{true} payoff when the mechanism is run on reports $(\hat s_i,s_{-i})$.
\end{theorem}

\begin{proof}
We prove the claim for Model~1; Model~2 is proved in Appendix~\ref{app:truthfulness-continuous}.

Fix an agent $i$, fix truthful reports $s_{-i}$ of all other agents, and consider any misreport $\hat s_i$.
Let $X^{max}$ be the collection profile output by Algorithm~\ref{alg:max} under $(s_i,s_{-i})$ and let $X'$ be the output under
$(\hat s_i,s_{-i})$.  

\textbf{Case 1: $x_i'\ge x_i^{\max}$ (misreport makes $i$ weakly higher).}
Let $\tau_i$ be the iteration that $i$ is truthfully selected by Algorithm~\ref{alg:max}. Due to prefix-independence and the fact that $i$ is moved later by the misreport, all agents fixed before $\tau_i$ are assigned identical collections in the two runs. Modify $X'$ such that the $K = n - \tau_i + 1$ agents fixed after $\tau_i$ in the truthful run are weakly above $x_i'$ to get $\tilde{X}$, by positive spillovers this can only improve $i$'s utility. 

By the iteration-highest-response property, we know that $i$ does not want to deviate upwards from $x_i^{\max}$ in $\tilde{X}$. So we can apply Lemma~\ref{lem:upward-harm-fixed-lower-plusone} to the pair $(x'_i,\tilde{X})$ and $(x^{\max}_i, \tilde{X})$ to get that $U(x'_i,\tilde{X}) < U_i(x^{\max}_i, \tilde{X})$. We also know that $U_i(x^{\max}_i, \tilde{X}) = U_i(X^{\max})$ since both the data collection and the total accessible data in the two cases is the same. Combining, we get

\[
U_i(X')
\;=\;
U_i(x'_i,\widetilde X_{-i})
\;<\;
U_i(x_i^{\max},\widetilde X_{-i})
\;=\;
U_i(X^{\max}).
\]

\textbf{Case 2: $x_i'\le x_i$ (misreport makes $i$ weakly lower).}
Apply Lemma~\ref{lem:downward-harm-robust} to the pair $(x_i, X_{-i})$ and $(x_i', X_{-i}')$. The prefix-independence statement gives the required hypothesis that all agents $j\neq i$ with $x_j < x'_i$ are fixed across $X$ and $X'$. Since $X$ is a Nash equilibrium, $i$ has no profitable local downward deviation at $x_i$. Therefore Lemma~\ref{lem:downward-harm-robust} implies $U_i(X')\le U_i(X)$.

In both cases, $U_i(X')\le U_i(X)$, so no misreport improves $i$'s true payoff. Hence truthful reporting is a Nash
equilibrium.
\end{proof}

\paragraph{Lack of truthfulness without enforcement.}
Model~3 is \emph{not} truthful: an agent can understate her willingness to collect data to induce others to collect more, and then, during the exchange stage, collect more than reported and benefit from the others’ inflated collections. Appendix~\ref{app:truthfulness-continuous} presents an explicit two-agent counterexample.

\subsection{Global Pareto Optimality}

In general supermodular games with positive spillovers, the maximal equilibrium is Pareto-best \emph{among equilibria}, but may be strictly Pareto dominated by off-equilibrium profiles \cite{levin2003supermodular}. A stronger result holds for the fair-exchange game.

\begin{theorem}[Global Pareto optimality]
\label{thm:global-pareto}
As usual, let $T^{max}$ denote the output of Algorithm ~\ref{alg:max}. No other total-data vector $T'$ achieves
\[
U_i(T') \ge U_i(T^{max})\ \text{for all } i
\quad\text{and}\quad
U_j(T') > U_j(T^{max})\ \text{for some}\ j.
\]
Thus the maximal fair-exchange equilibrium is Pareto-optimal over all strategy profiles.
\end{theorem}

\begin{proof}
Assume for contradiction that there exists a profile $X'$ that Pareto-dominates the maximal equilibrium $X^{\max}$.

\textbf{Step 1: Not all agents can move weakly down.}
If $X' \le X^{\max}$ coordinatewise, then for each agent $j$,
\[
U_j(x'_j,X'_{-j}) \;\le\; U_j(x'_j,X^{\max}_{-j}) \;\le\; U_j(x^{\max}_j,X^{\max}_{-j}),
\]
where the first inequality is positive spillovers and the second uses that $X^{\max}$ is a Nash equilibrium.
This contradicts strict Pareto domination. Hence there exists some agent with $x'_i > x^{\max}_i$.

\textbf{Step 2: Choose the first agent (in Algorithm~\ref{alg:max} order) who moves up.}
Let $(1),(2),\dots,(n)$ denote the selection order of Algorithm~\ref{alg:max} (so $(\tau)$ is the agent selected in
iteration $\tau$), and let
\[
\tau^\star \;:=\; \min\bigl\{\tau : x'_{(\tau)} > x^{\max}_{(\tau)}\bigr\},\qquad i := (\tau^\star).
\]

\textbf{Step 3: Restore lower agents to equilibrium}
Define $\widetilde X$ by $\widetilde x_i := x'_i$, by $\widetilde x_j := x^{\max}_j$ for every $j\neq i$ with
$x^{\max}_j < x^{\max}_i$, and by $\widetilde x_y := x'_y$ for all remaining agents $y$.
Then $\widetilde X_{-i} \ge X'_{-i}$ coordinatewise, so by positive spillovers,
$U_i(X') \;\le\; U_i(\widetilde X)$.

\textbf{Step 4: agent $i$ is strictly worse off.}
Consider the profile $(x_i^{\max},\widetilde X_{-i})$ in which agent $i$ plays its equilibrium action $x_i^{\max}$
against opponents $\widetilde X_{-i}$. By construction of $\widetilde X$, every agent $j\neq i$ with
$x^{\max}_j < x^{\max}_i$ is fixed at $x^{\max}_j$ in $(x_i^{\max},\widetilde X_{-i})$.

Due to the iteration-highest-response property, we know that $i$ does not want to deviate upwards from $(x_i^{\max},\widetilde X_{-i})$.

Therefore we may apply Lemma~\ref{lem:upward-harm-fixed-lower-plusone} to the pair of profiles $(x_i^{\max},\widetilde X_{-i})$ and $(x'_i,\widetilde X_{-i})$ to get $U_i(x'_i,\widetilde X_{-i}) \;<\; U_i(x_i^{\max},\widetilde X_{-i})$.
Also note that
$U_i(x_i^{\max},\widetilde X_{-i}) \;=\; U_i(x_i^{\max},X^{\max}_{-i})$ since both data collection costs and total accessible data are the same for $i$ in these two cases. 
Combining these yields: 
\[
U_i(\widetilde X)
\;=\;
U_i(x'_i,\widetilde X_{-i})
\;<\;
U_i(x_i^{\max},\widetilde X_{-i})
\;=\;
U_i(X^{\max}).
\]
Together with $U_i(X')\le U_i(\widetilde X)$ from Step~3, this contradicts that $X'$ Pareto-dominates $X^{\max}$.
\end{proof}

The maximal fair-exchange equilibrium simultaneously admits truthful implementation under realistic enforcement and is globally Pareto-optimal in the continuous model. This is an unusually strong combination of incentive and welfare guarantees that emerges because of the game’s highly structured exchange geometry.

\section{Graph-Restricted Fair Exchange}
\label{sec:graph}

We now study fair exchange when data sharing is restricted by a fixed undirected graph $G = (V,E)$. Each agent may only exchange data with its neighbors: agent $i$ receives shared data only from agents in $N(i)$. Missing edges in the graph capture situations where agents face legal, ethical, or competitive barriers to sharing data. For example, two competing organizations may categorically refuse to exchange information, or universities and non-profits may be legally barred from sharing data with for-profit entities.

In the complete-graph model, supermodularity in total-data space guarantees a clean lattice structure of equilibria. This property does not hold for general graphs. The graph structure means that a non-neighboring agent increasing their collection level may actually hurt an agent in terms of utility, because it may reduce the collection level of their neighbor. As a consequence, the fair-exchange game on graphs is not supermodular, best responses are not monotone, and a Pareto-best equilibrium may not exist. Appendix~\ref{subsec:graph-super-counter} provides explicit examples illustrating these failures.

Despite this structural difficulty, we show that an analog of the continuous maximal-equilibrium algorithm continues to work on graphs: it produces a pure Nash equilibrium and, among all such equilibria obtainable by this construction, yields a Pareto-undominated outcome.

\subsection{Graph Restricted Model}

Under graph-restricted exchange, agent $i$’s total accessible data is
\[
t_i(x) = x_i + \sum_{j \in N(i)} \min\{x_i, x_j\}.
\]
Because only neighbors contribute to $t_i$, optimality conditions become \emph{local}: $i$'s incentives depend on its own $x_i$ and on the contributions of its neighbors. If $i$ increases $x_i$, only edges incident to $i$ are affected. This locality replaces the total-order structure used in the complete-graph construction.

Our key observation is that even though the game is not supermodular, incentive comparisons still have a monotonic structure \emph{within neighborhoods}: when neighbors' contributions rise, $i$'s marginal value of contributing also weakly rises. This property is weaker than global supermodularity but strong enough to support a constructive equilibrium algorithm.

\subsection{Optimal Data Equilibrium Algorithm for Graph-Restricted Case}

The algorithm mirrors the ideas from the continuous maximal-equilibrium construction by iteratively fixing agents data collection levels, from smallest to largest. However, due to the graph restriction, we operate in data collection space, since agents are not oblivious to which particular earlier agent collected data, it matters whether the earlier agent was a neighbor or not. 

The definition of the $K$-data-levels does not change from the continuous case. As input, we only need them for the degree of an agent plus one, rather than $n$. This is because, due to graph-restriction, $d_i(G) + 1$ the maximum number of agents that $i$ can actually exchange with (including themselves).

Similar to the continuous case, at each step, the algorithm identifies the remaining agent who should collect the least data. This is done by computing what each agent would collect under maximally favorable conditions (when all remaining agents collect more than them) and selecting the agent with the lowest level. It then locks this lowest agent in to that level (subject to a nondecreasing floor) and removes the agent from the active graph. The additional step we need for the graph case is to edit the remaining players "residual" $K$-data-levels to capture the fact that agents only exchange with their neighbors. 

\begin{algorithm}[H]
\label{alg:graph}
  \SetAlgoNoLine
  \KwIn{Graph $G=(V,E)$; for each $i\in V$, $K$-data-levels $\{s_i^{\,K}\}_{K=1}^{d_i(G) + 1}$}
  \KwOut{Contribution profile $x$}

  $H \gets G$\;
  \ForEach{$i \in V$}{
    $r_i^\bullet \gets s_i^\bullet$\tcp*[r]{residual threshold levels}
  }
  \texttt{prev} $\gets 0$\;

  \While{$V(H) \neq \emptyset$}{
    Let $i \in V(H)$ minimize
    $\rho_i := \dfrac{r_i^{\,d_i(H)+1}}{d_i(H)+1}$ (break ties deterministically)\;

    $x_i \gets \max\{\texttt{prev},\, \rho_i\}$\;
    \texttt{prev} $\gets x_i$\;

    \ForEach{$j \in N_H(i)$}{
      $r_j^\bullet \gets r_j^\bullet - x_i$\;
    }

    Delete $i$ (and its incident edges) from $H$\;
  }

  \caption{Graph-Restricted Optimal Data Equilibrium}
\end{algorithm}

\subsection{Properties of the Graph-Restricted Equilibrium}

The profile $X^{opt}$ produced by Algorithm~\ref{alg:graph} is a pure Nash equilibrium of the graph-restricted fair-exchange game and, while it is not maximal the same sense, it has many properties similar to the maximal equilibrium in the continuous setting. For the mechanism design problem analogous to the continuous case, enforcing $X^{opt}$ incentivizes truthful reporting of the $K$-data-levels. In addition, $X^{opt}$ is globally Pareto optimal, no profile (including off-equilibrium profiles) dominates it. The proofs of these claims tightly follow the proofs of Theorem~\ref{thm:truthful} and Theorem~\ref{thm:global-pareto} from the continuous case. Further details and proofs for the graph case can be found in Appendix ~\ref{app:graph}.

\section{Discrete Fair Exchange}
\label{sec:discrete}
While real-world data collection is almost always discrete, it is common when studying data exchange and federated learning to treat data as infinitely divisible. This is because the amounts of data are large enough that the discretization does not have large effects. However, for completeness, we include a study of the case where agents' collection/contribution decisions are restricted to integer units, representing a setting where relatively few samples are collected (compared to the number of agents).

Similar to the graph case, supermodularity fails in the discrete setting, even in the induced game in total-data space. Again similar to the graph case, there may not exist a Pareto-best equilibrium. We give examples illustrating these facts in Appendix~\ref{subsec:discrete-super-counter}.

Again like to the graph case, an analog of Algorithm~\ref{alg:max} works in the discrete case as well. We can compute a discrete optimal data equilibrium through a rounding and verification algorithm that fixes agents from lowest data collection level to the highest. At a high level, instead of looking for individual remaining agent that wants to collect the least data, we select the whole group of agents whose rounded-down desired collection level, denoted $m$, is the same. We consider if anyone in this group would deviate downwards if the whole group was placed at $m+1$. If not, we can fix the whole group at $m$ and if so, we can place the deviator at $m$ and continue. 

We show that the resulting profile is an \emph{exact} Nash equilibrium of the discrete game. We follow the continuous case proofs once again to argue that this profile leads to a truthful mechanism and is once again Pareto undominated among \emph{all} strategy profiles. All details and proofs for the discrete case can be found in Appendix ~\ref{app:discrete}.

\section{Discussion and Open Questions}
\label{sec:discussion}

Our results show that fair-exchange is a well-motivated contract for the simple setting where agents have the ability to sample (for a cost) from a single data pool. It can deliver truthful and Pareto-efficient data sharing, with clean structure in the continuous model and robust extensions to discrete and graph-restricted settings. At the same time, the model adopts strong simplifications about the nature of data, and relaxing these assumptions raises several important questions.

First, the model treats all data points as interchangeable. In many applications, agents collect different \emph{types} of data—clinical records, genomic sequences, behavioral traces—whose values interact in complementary or substitutable ways. Extending fair exchange to multi-type data could change both equilibrium structure and the form of truthful mechanisms.

Second, we assume that agents' samples do not overlap. This is reasonable when datasets are large, but unrealistic when populations are small or highly specialized. Overlap and correlation introduce redundancy and may alter incentives in ways that break the clean monotonicity structure underlying our equilibrium characterizations. Understanding fair exchange when marginal values depend on overlaps remains open.

Finally, our mechanism assumes agents submit genuine data. In many real-world settings, agents may fabricate, filter, or manipulate data in ways that are difficult to detect. Incorporating partial verifiability or statistical validation into the mechanism—and understanding whether truthful reporting can still be sustained—presents a challenging direction for future work.

In addition, we take a mechanism design approach to arguing that equilibria are a good solution concept for the data-exchange game. Another, more decentralized, approach is to consider the convergence of learning algorithms where players aim to learn their optimal data collection levels over time. It is known that a wide class of learning algorithms converge to equilibria in supermodular games \cite{milgromroberts90}, so this is an appealing direction. 

While fair exchange offers a simple and compelling benchmark, extending the analysis to richer forms of data—heterogeneous, overlapping, or strategically manipulable—and to alternative strategic environments remains an important open direction. 

\begin{acks}

	The authors gratefully acknowledge support from NSF grants CCF-2212233 and CCF-2332922 and from a MURI grant by ONR.

\end{acks}
\newpage
\bibliographystyle{ACM-Reference-Format}
\bibliography{bibliography}

\appendix
\input{arxiv_appendix.tex}
\end{document}

%% file: arxiv_appendix.tex
\newpage
\section{Continuous Model}
\subsection{Forward and Inverse Fair-Exchange Transforms}
\label{fwdinv}
We write $T=\Phi(X)$ for the induced total accessible data profile, and $X=\Phi^{-1}(T)$ for the inverse map, defined on the image of $\Phi$.

\begin{algorithm}
  \SetAlgoNoLine
  \KwIn{Collection profile $X=(x_1,\dots,x_n)\in \mathbb{R}^n_{\ge 0}$}
  \KwOut{Total-data profile $T=(t_1,\dots,t_n)\in \mathbb{R}^n_{\ge 0}$}

  \For{$i \gets 1$ \KwTo $n$}{
    $t_i \gets x_i$\;
    \For{$j \gets 1$ \KwTo $n$}{
      \If{$j \ne i$}{
        $t_i \gets t_i + \min\{x_i,x_j\}$\;
      }
    }
  }

  \caption{Forward Transform $\Phi(X)$}
  \label{alg:phi-fwd}
\end{algorithm}

\begin{algorithm}
  \SetAlgoNoLine
  \KwIn{Total-data profile $T=(t_1,\dots,t_n)\in \mathbb{R}^n_{\ge 0}$}
  \KwOut{Collection profile $X=(x_1,\dots,x_n)\in \mathbb{R}^n_{\ge 0}$ such that $\Phi(X)=T$, if it exists}

  Let $(1),\dots,(n)$ be an ordering of agents such that
  $t_{(1)} \le t_{(2)} \le \cdots \le t_{(n)}$\;
  $P \gets 0$\tcp*[r]{prefix sum of recovered collections}

  \For{$\tau \gets 1$ \KwTo $n$}{
    $x_{(\tau)} \gets \dfrac{t_{(\tau)} - P}{\,n-\tau+1\,}$\;
    $P \gets P + x_{(\tau)}$\;
  }
  \Return{$X$}\;
  \caption{Inverse transform $\Phi^{-1}(T)$}
  \label{alg:inverse}
\end{algorithm}

\begin{theorem}[Correctness of the inverse transform]
\label{thm:inverse-correct}
For any $X\in\mathbb{R}_{\ge 0}^n$, let $T=\Phi(X)$. Then Algorithm~\ref{alg:inverse}
returns the same $X$, up to permutation within tied coordinates.
\end{theorem}

\begin{proof}
Let $x_{(1)} \le x_{(2)} \le \cdots \le x_{(n)}$ be the coordinates of $X$ in nondecreasing order,
and let $t_{(\tau)}$ denote the corresponding totals for those same agents, i.e.\ $t_{(\tau)}=(\Phi(X))_{(\tau)}$
under the same tie-breaking.

For each rank $\tau$,
\[
t_{(\tau)}
= \sum_{j<\tau} x_{(j)} + (n-\tau+1)\,x_{(\tau)}.
\tag{A.1}
\label{A.1}
\]
Moreover,
\[
t_{(\tau+1)}-t_{(\tau)} = (n-\tau)\big(x_{(\tau+1)}-x_{(\tau)}\big)\ge 0,
\]
so sorting $T$ recovers the same rank order as sorting $X$, up to ties.

Now run Algorithm~\ref{alg:inverse} on the sorted totals $t_{(1)}\le \cdots \le t_{(n)}$.
Maintaining $P=\sum_{j<\tau} x_{(j)}$, rearranging \eqref{A.1} gives
\[
x_{(\tau)} = \frac{t_{(\tau)} - \sum_{j<\tau} x_{(j)}}{n-\tau+1}
          = \frac{t_{(\tau)} - P}{n-\tau+1},
\]
which is exactly the algorithm’s update. Thus the algorithm reconstructs all $x_{(\tau)}$ in order,
and mapping back to agent identities yields the original $X$ (with arbitrary tie-breaking within equal blocks).
\end{proof}

\subsection{Not Supermodular in Collection Space}

Consider a game with 2 agents $p_1, p_2$. Data costs for $p_1$ are 1 and data benefit function is:
\[
b(T) =
\begin{cases}
(1+\epsilon)T, & T \leq 10, \\[6pt]
10(1+\epsilon), & T > 10.
\end{cases}
\]

This setup implies that $p_1$ best response is to collect data such that their total accessible data is 10. If $p_2$ collects 5 then $p_1$'s best response is to collect 5. If $p_2$ collects 0 then $p_1$'s best response is to collect 10. This violates increasing best responses, so this framing is not supermodular.

\subsection{Minimal Equilibrium Algorithm}
To compute the minimal data equilibrium, we must define a variant of $K$-data-levels called min-$K$-data levels which represent the least amount of data agents are willing to collect under certain circumstances. The upwards and downwards constraints are the same with these min-$K$-data levels, they only differ from $K$-data-levels due to agent's \emph{indifference}.
\begin{definition}[min-$K$-data level]
\label{def:minkdatalevel}
For each agent $i$ and integer $K \ge 0$, the \emph{min-$K$-data level}, denoted $\tilde{s_i}^K$, is defined by
\[
\tilde{s_i}^K \;:=\; \min\Bigl\{\,s \ge 0 \;:\; b_i'(s)\leq\tfrac{c_i}{K}\Bigr\}.
\]
\end{definition}

\label{app:min}
\begin{algorithm}[H]
  \SetAlgoNoLine
  \KwIn{min-$K$-data-levels $\{\tilde{s}_{j}^{\,K}\}_{j\in[n],\,K\in\{1,\dots,n\}}$}
  \KwOut{Total-data equilibrium vector $T \in \mathbb{R}^n_{\ge 0}$}

  $R \gets \{1,\dots,n\}$ \tcp*[r]{remaining agents}
  $\textit{prev} \gets +\infty$\;

  \For{$K \gets 1$ \KwTo $n$}{
    select $j \in R$ maximizing $\tilde{s}_{j}^{\,K}$\;
    $t_{j} \gets \min\{\,\tilde{s}_{j}^{\,K},\ \textit{prev}\,\}$\;
    $\textit{prev} \gets t_{j}$\;
    $R \gets R \setminus \{j\}$\;
  }
  \caption{Minimal Data Equilibrium}
  \label{alg:min}
\end{algorithm}
Let $T^{\min}$ be the output of Algorithm~\ref{alg:min} and let $X^{\min}=\Phi^{-1}(T^{\min})$
\begin{theorem}[Correctness of Algorithm~\ref{alg:min}]
\label{thm:min-alg-ne}
$X^{\min}$ is a Nash equilibrium
\end{theorem}

\begin{proof}[Proof (parallel to proof of Theorem~\ref{thm:max-alg-ne})]
Rename agents by the selection order of Algorithm~\ref{alg:min} as $(1),(2),\dots,(n)$, where $(K)$ is
the agent selected in iteration $K$. The algorithm maintains a running \emph{upper bound} $\textit{prev}$
and assigns
\[
T^{\min}_{(K)}=\min\{\tilde{s}_{(K)}^{\,K},\textit{prev}\},
\qquad\text{so}\qquad
T^{\min}_{(1)}\ge T^{\min}_{(2)}\ge \cdots \ge T^{\min}_{(n)}.
\]

We verify the two local equilibrium inequalities for each agent $(K)$, exactly as in the maximal case.

\emph{No-downward deviation.} By construction $T^{\min}_{(K)}\le \tilde{s}_{(K)}^{\,K}$.
In the final profile, at least the $K$ agents $(1),\dots,(K)$ are weakly above $(K)$ (since the assigned
levels are nonincreasing), hence $k_{(K)}(T^{\min})\ge K$. Monotonicity of $\tilde{s}_i^{\,\cdot}$ in $K$
then gives
\[
T^{\min}_{(K)}\le \tilde{s}_{(K)}^{\,K}\le \tilde{s}_{(K)}^{\,k_{(K)}(T^{\min})},
\]
which is the no-downward inequality.

\emph{No-upward deviation.} Only earlier-selected agents can be strictly above $(K)$, so
$k^{\uparrow}_{(K)}(T^{\min})\le K$. If $\textit{prev}$ does not bind, then
$T^{\min}_{(K)}=\tilde{s}_{(K)}^{\,K}$ and monotonicity directly implies
$T^{\min}_{(K)}\ge \tilde{s}_{(K)}^{\,k^{\uparrow}_{(K)}(T^{\min})}$.
If $\textit{prev}$ binds (so $T^{\min}_{(K)}=T^{\min}_{(K-1)}$), then $(K)$ is assigned the same level as
the previous selected agent. Since $(K-1)$ was selected to maximize $\tilde{s}_j^{\,K-1}$ among remaining agents, we have $\tilde{s}_{(K-1)}^{\,K-1}\ge \tilde{s}_{(K)}^{\,K-1}$, and combining this with the fact that $(K-1)$ satisfies its own no-upward constraint yields
$T^{\min}_{(K)}=T^{\min}_{(K-1)}\ge \tilde{s}_{(K)}^{\,K-1}\ge \tilde{s}_{(K)}^{\,k^{\uparrow}_{(K)}(T^{\min})}$.

Thus every agent satisfies both inequalities, so $T^{\min}$ (and hence $X^{\min}$) is a Nash equilibrium.
\end{proof}

\begin{theorem}[Extremality of the minimal equilibrium]
\label{thm:min-alg-extremal}
In $X^{\min}$, each agent has weakly less total accessible data than they do in any other equilibrium. In other words, no other equilibrium $X'$ exists where $T' = \Phi(X')$ is smaller than $T^{\min}$ in any component.
\end{theorem}

\begin{proof}[Proof sketch (parallel to proof of Theorem~\ref{thm:max-alg-extremal})]
Rename agents by the selection order $(1),\dots,(n)$ of Algorithm~\ref{alg:min}.
Suppose for contradiction that there exists an equilibrium $T'$ with $T'_{i}<T^{\min}_{i}$ for some agent.
Let $K$ be the first index such that $T'_{(K)}<T^{\min}_{(K)}$.
Then for every $r<K$,
\[
T'_{(r)}\ge T^{\min}_{(r)}\ge T^{\min}_{(K)} > T'_{(K)},
\]
so all earlier agents are strictly above $(K)$ in $T'$, and hence $k^{\uparrow}_{(K)}(T')\ge K$.
By monotonicity of $\tilde{s}_{(K)}^{\,\cdot}$ we get
$\tilde{s}_{(K)}^{\,k^{\uparrow}_{(K)}(T')}\ge \tilde{s}_{(K)}^{\,K}$.
But Algorithm~\ref{alg:min} assigns $T^{\min}_{(K)}=\min\{\tilde{s}_{(K)}^{\,K},\textit{prev}\}\le \tilde{s}_{(K)}^{\,K}$,
so
\[
T'_{(K)}<T^{\min}_{(K)}\le \tilde{s}_{(K)}^{\,K}\le \tilde{s}_{(K)}^{\,k^{\uparrow}_{(K)}(T')},
\]
which violates the no-upward-deviation inequality for $(K)$ in $T'$.
Contradiction. Therefore no equilibrium can lie strictly below $T^{\min}$ in any coordinate.
\end{proof}
\section{Truthfulness}
\label{app:truthfulness-continuous}

This appendix proves the truthfulness claims for Models~2 in the Continuous setting, and provides a counterexample for Model~3.

\subsection{Model 2 (outside collection, non-exchangeable)}
\label{app:model2-continuous}

\begin{lemma}[Outside-option closure]
\label{lem:outside-closure}
Define
\[
v_i(t)\;=\;\max_{z\ge 0}\big(u_i(t+z)-c_i z\big).
\]
Then $v_i$ is concave and nondecreasing.
\end{lemma}

\begin{proof}
Let $g(t,z)=u_i(t+z)-c_i z$. Since $u_i$ is concave and $t+z$ is affine, $g$ is jointly concave.

Fix $t\ge 0$. Because $b_i'(t)\to 0$ as $t\to\infty$ and $c_i>0$, we have
$\frac{d}{dz}g(t,z)=b_i'(t+z)-c_i<0$ for all sufficiently large $z$, so the maximum over $z\ge 0$ is attained.

Concavity: for any $t_1,t_2\ge 0$, $\lambda\in[0,1]$, and any $z_1,z_2\ge 0$,
\[
g(\lambda t_1+(1-\lambda)t_2,\ \lambda z_1+(1-\lambda)z_2)
\ge
\lambda g(t_1,z_1)+(1-\lambda)g(t_2,z_2).
\]
Taking $\max_{z_1,z_2\ge 0}$ on the right and using $\lambda z_1+(1-\lambda)z_2\ge 0$ yields
\[
v_i(\lambda t_1+(1-\lambda)t_2)\ge \lambda v_i(t_1)+(1-\lambda)v_i(t_2).
\]
Monotonicity: if $t'\ge t$ then $g(t',z)\ge g(t,z)$ for all $z$ since $b_i$ is nondecreasing, hence
$v_i(t')\ge v_i(t)$.
\end{proof}

\begin{lemma}[Model 2 payoff representation]
\label{lem:model2-rep}
Fix any report profile $\hat s$, and let $(x(\hat s),T(\hat s))$ denote the mechanism's on-platform outcome.
In Model~2, after optimizing outside collection, agent $i$'s realized payoff equals
\[
U_i^{M2}(\hat s)\;=\;v_i(T_i(\hat s))\;-\;c_i x_i(\hat s).
\]
\end{lemma}

\begin{proof}
Under Model~2, submitting less than the required amount yields no data from others, and therefore cannot be better than submitting the required amount and then (if desired) discarding received data, since $u_i$ is nondecreasing.
Submitting more than the required amount does not change exchange outcomes (outside data is non-exchangeable) and does not affect cost, so we can ignore this case. 
Hence the on-platform contribution is effectively pinned to $x_i(\hat s)$, and the only remaining decision is how much outside data $z\ge 0$ to collect privately, which adds $z$ to $i$'s accessible total and costs $c_i z$.
Optimizing over $z$ yields $v_i(T_i(\hat s))$, giving the stated payoff.
\end{proof}

\begin{proposition}[Model 2 truthfulness in the Continuous setting]
\label{prop:model2-truthful-cont}
In the Continuous setting under Model~2, truthful reporting is a Nash equilibrium.
\end{proposition}

\begin{proof}
By Lemma~\ref{lem:model2-rep}, the reporting-stage objective under Model~2 is $v_i(T_i(\hat s)) - c_i x_i(\hat s)$.
By Lemma~\ref{lem:outside-closure}, $v_i$ satisfies the same concavity/monotonicity assumptions used in the proof of Theorem~\ref{thm:truthful}.
Therefore the argument of Theorem~\ref{thm:truthful} applies verbatim with $u_i$ replaced by $v_i$, implying no profitable misreport.
\end{proof}

\subsection{Model 3 (no enforcement): counterexample}
\label{app:model3-counterexample}

Note that an agent's strategy is two-phase - first to report their $k$-data-levels honestly and second to play what the recommendation suggests. We denote this full strategy as truthful. Truthful strategies do not form an NE in Model 3 as they do in Models 1 and 2.

\begin{proof}
Consider two agents with costs $c_1=c_2=1$ and step-value utilities
\[
u_1(t)=100\cdot \mathbf{1}\{t\ge 10\},
\qquad
u_2(t)=20\cdot \mathbf{1}\{t\ge 8\}.
\]
Under truthful reporting, the mechanism recommends $(x_1,x_2)=(6,4)$.
Under a misreport by agent~1 claiming it derives no benefit from data, the mechanism recommends $(x_1,x_2)=(0,8)$.

If agent~1 reports truthfully and then plays the recommendation, its payoff is $100-6=94$.
If agent~1 instead misreports so that the recommendation becomes $(0,8)$, then by assumption agent~2 plays $x_2=8$.
Agent~1 (the deviator) then plays its true best response to $x_2=8$, which is $x_1=5$, achieving its threshold while paying cost $5$, for payoff $100-5=95$.
Since $95>94$, truthful reporting and playing the recommendation is not a best response for agent~1.
\end{proof}

\section{Graph Case}
\label{app:graph}
We define new neighbor-based rankings functions $\kappa$, analogous to Section~\ref{def:rankings}, that capture the relevant quantity in the graph case.
\[
\kappa_i(x)\;:=\;1+\bigl|\{\,j\in N_G(i):\; x_j \ge x_i\,\}\bigr|
\;
\qquad\text{(\# neighbors weakly above, plus self),}
\]
\[
\kappa_i^{\uparrow}(x)\;:=\;1+\bigl|\{\,j\in N_G(i):\; x_j > x_i\,\}\bigr|
\;
\qquad\text{(\# neighbors strictly above, plus self).}
\]

We use these neighbor-based rankings to define the local no deviation conditions in the graph case. Due to Lemmas~\ref{lem:upward-harm-fixed-lower-plusone} and  \ref{lem:downward-harm-robust} still holding in the graph case (the proofs go through unchanged), if these local conditions hold for all agents, a profile is a Nash equilibrium.

\noindent\textbf{No profitable upward deviation:}
$t_i(x^\ast)\ge s_i^{\,\kappa_i^{\uparrow}(x^\ast)}$ for all $i\in V$.

\noindent\textbf{No profitable downward deviation:}
$t_i(x^\ast)\le s_i^{\,\kappa_i(x^\ast)}$ for all $i\in V$.

\subsection{Graph Supermodularity Counterexamples}
\label{subsec:graph-super-counter}
In the graph-restricted model, the limitations on who can exchange data break both supermodularity in total data space and imply that there is no Pareto-best equilibrium (equilibria may be incomparable in terms of agent utilities). We give two examples that illustrate these phenomena.

\begin{example}[Non monotone BR in the graph game (total data space)]
\label{ex:graph-nonmonotone-br}
3 agents, $p_1, p_2, p_3$ on a graph. 
\begin{center}
\includegraphics[width=0.3\linewidth]{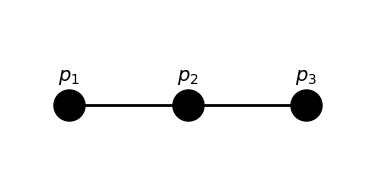}
\end{center}
Consider $p_1$ with data costs 1 and data benefit function (with $\epsilon$ small):
\[
b(T) =
\begin{cases}
T, & T \leq 14, \\
14+\frac{2T}{3}, & T > 14, \\
14+\frac{32}{3}, & T > 16, 
\end{cases}
\]
In the first case, agents 2 and 3 collect $(8, 5)$. Then, the best response for $p_1$ is to collect $8$ and get total data $16$, since their neighbor is collecting $8$. So the total data vector is $(16, 16, 10)$

In the second case, agents 2 and 3 collect $(7, 93)$. Then, the best response for $p_1$ is to collect $7$ and get total data $14$. So the total data vector is $(14, 21, 100)$. So the best responses are not monotone in total data space (an agent does not necessarily access more data when their opponents access more data). 
\end{example}

\begin{example}[No Pareto-best equilibrium]
\label{ex:graph-incomparable-eq}
5 players $p_1, p_2, p_3, p_4, p_5$ arranged on a graph.
\begin{figure}[H]
  \centering
  \includegraphics[width=0.3\linewidth]{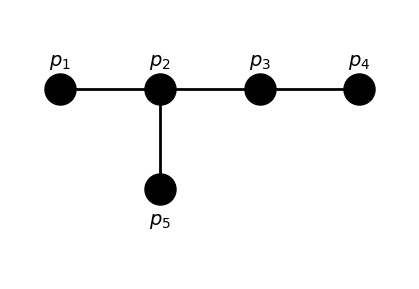}
\end{figure}
$K$-data-levels according to the below table.
\begin{center}
\begin{tabular}{c|ccccc}
\hline
$K$ & $p_1$ & $p_2$ & $p_3$ & $p_4$ & $p_5$ \\
\hline
1 & 0   & 10  & 10  & 21 & 15 \\
2 & 100 & 100 & 10  & 21 & 16 \\
3 & -   & -   & 157 & 21 & -  \\
4 & -   & -   & 157 & -  & -  \\
\hline
\end{tabular}
\end{center}
Then $X_1 = (5, 5, 8, 8, 5)$ and $X_2 = (50, 50, 7, 8, 50)$ are incomparable equilibria. The corresponding total data accessible profiles are $T_1 = (10, 10, 21, 16, 10)$ and $T_2 =(100, 157, 21, 15, 100)$. We can observe that $p_3$ is strictly better off in $X_2$ (collects less and total accessible is the same) and $p_4$ is strictly worse off in $X_2$ (collects same and total accessible is less). There is no equilibrium where both achieve their maximal utility and therefore no Pareto-best equilibrium.

\end{example}

\subsection{Graph Algorithm Correctness}
\begin{theorem}[Correctness of Algorithm~\ref{alg:graph}]
\label{thm:graph-alg-ne}
The profile $X^{opt}$ returned by Algorithm~\ref{alg:graph} is a Nash equilibrium.
\end{theorem}

\begin{proof}
Recall that the algorithm maintains a running floor \texttt{prev}, so assignments are
nondecreasing in deletion order.

\emph{No upward deviation.}
Fix a vertex $i$, and let $H$ be the residual graph right before $i$ is deleted.
By monotonicity, every later neighbor $j\in N_H(i)$ satisfies $x_j^{\mathrm{opt}}\ge x_i^{\mathrm{opt}}$ and every earlier
neighbor $u\in N_G(i)\setminus V(H)$ satisfies $x_u^{\mathrm{opt}}\le x_i^{\mathrm{opt}}$. Hence in the final profile $X^{\mathrm{opt}}$,
\[
t_i(X^{\mathrm{opt}})
= x_i^{\mathrm{opt}}+\sum_{u\in N_G(i)\setminus V(H)}x_u^{\mathrm{opt} }+\sum_{j\in N_H(i)}\min\{x_i^{\mathrm{opt}},x_j^{\mathrm{opt}}\}
=\sum_{u\in N_G(i)\setminus V(H)}x_u^{\mathrm{opt}}+(d_i(H)+1)x_i^{\mathrm{opt}} .
\]
By construction of the residual thresholds, $r_i^K$ equals $s_i^K$ minus the total contribution
of already-fixed neighbors; in particular,
\[
s_i^{\,d_i(H)+1}
=\sum_{u\in N_G(i)\setminus V(H)}x_u^{\mathrm{opt}} + r_i^{\,d_i(H)+1}.
\]
Therefore $t_i(X^{\mathrm{opt}})\ge s_i^{\,d_i(H)+1}$ is equivalent to
$(d_i(H)+1)x_i^{\mathrm{opt}}\ge r_i^{\,d_i(H)+1}$, i.e.\ $x_i^{\mathrm{opt}}\ge \rho_i$, where
$\rho_i := r_i^{\,d_i(H)+1}/(d_i(H)+1)$. Since the algorithm sets
$x_i^{\mathrm{opt}}=\max\{\texttt{prev},\rho_i\}$, we have $x_i^{\mathrm{opt}}\ge \rho_i$ and hence $t_i(X^{\mathrm{opt}})\ge s_i^{\,d_i(H)+1}$.
Finally, $\kappa_i^{\uparrow}(X^{\mathrm{opt}})\le d_i(H)+1$ and $s_i^K$ is nondecreasing in $K$, so
$t_i(X^{\mathrm{opt}})\ge s_i^{\,\kappa_i^{\uparrow}(X^{\mathrm{opt}})}$. Thus no vertex has a profitable upward deviation.

\emph{No downward deviation.} Fix a vertex $i$ again and consider two cases.
 
First, we consider the care where $\rho_i>\texttt{prev}$, so $x_i^{\mathrm{opt}}=\rho_i$. Then $x_i^{\mathrm{opt}}$ is strictly larger than all earlier assigned levels, so no earlier neighbor can tie at $x_i^{\mathrm{opt}}$.
Thus $\kappa_i(X^{\mathrm{opt}})=d_i(H)+1$. Moreover, $(d_i(H)+1)x_i^{\mathrm{opt}}=r_i^{\,d_i(H)+1}$, so the same algebra as above yields
\[
t_i(X^{\mathrm{opt}})=s_i^{\,d_i(H)+1}=s_i^{\,\kappa_i(X^{\mathrm{opt}})},
\]
and $i$ cannot profitably deviate downward.

Second, we consider the case where $\rho_i\le \texttt{prev}$, so $x_i^{\mathrm{opt}}=\texttt{prev}$. Let $\tau$ be the iteration when \texttt{prev} first becomes this value.
At time $\tau$, all already-deleted neighbors of $i$ are fixed forever, and all remaining neighbors
are eventually assigned levels $\ge \texttt{prev}$ (either exactly \texttt{prev} if the floor binds for them too, or larger by monotonicity). In the environment at time $\tau$ where $i$'s fixed neighbors are as already set and every other neighbor is treated as weakly above \texttt{prev}, vertex $i$ is willing to collect at least \texttt{prev}, i.e.\ it has no profitable downward deviation from \texttt{prev}. Since later steps only (weakly) increase neighbors' levels relative to that environment and never change the fixed neighbors, they cannot create a new profitable downward deviation for $i$. Hence $t_i(X^{\mathrm{opt}})\le s_i^{\,\kappa_i(X^{\mathrm{opt}})}$ also holds in the final profile $X^{\mathrm{opt}}$.

Since every vertex satisfies both the upward and downward local constraints at the output $X^{\mathrm{opt}}$, $X^{\mathrm{opt}}$ is a pure Nash equilibrium.
\end{proof}

\subsection{Graph Truthfulness and Pareto Optimality}
\subsubsection{Mechanism and reports}
Similar to the continuous case, we interpret the optimal-equilibrium construction as a direct mechanism. We assume that each agent knows their own cost and benefit function and can therefore privately compute their vector of $K$-data-levels (as in Definition~\ref{def:kdatalevel}). This is a sufficient summary of the benefit function for algorithmic purposes. Note that in the graph case, $K$-data-levels need to be reported for $K\in \{1,...,d_i(G)+1\}$ rather than all the way up to $n$ (of course, these coincide in the fully connected graph case).

\subsubsection{Proof Adaptation: Continuous to Graph}
\begin{theorem}[Graph truthfulness under enforcement]
\label{thm:graph-truthful}
Under Model~1, truthful reporting is a Nash equilibrium of the reporting game in the graph case: for all $i$, all $S=(s_i,s_{-i})$, and all misreports $\hat s_i$,
\[
U_i(S) \;\ge\; U_i(\hat s_i, s_{-i}),
\]
where $U_i(\hat s_i, s_{-i})$ denotes $i$'s \emph{true} payoff when the mechanism is run on reports $(\hat s_i,s_{-i})$.
\end{theorem}
\begin{theorem}[Graph global Pareto optimality]

\label{thm:graph-global-pareto}
$X^{opt}$ is globally Pareto optimal. No other data-collection vector $X'$ achieves
\[
U_i(X') \ge U_i(X^{opt})\ \text{for all } i
\quad\text{and}\quad
U_j(X') > U_j(X^{opt})\ \text{for some}\ j.
\]
Thus the optimal fair-exchange equilibrium for graphs is Pareto-optimal over all strategy profiles.
\end{theorem}

Rather than reproducing full proofs, we isolate the key properties needed to adapt the continuous proofs of Theorems~\ref{thm:truthful} and~\ref{thm:global-pareto} to the graph setting. Once these properties are established, the remainder of each argument proceeds exactly as in the continuous case.

In the graph case, positive spillovers in data collection space are still immediate. In addition, Lemmas \ref{lem:upward-harm-fixed-lower-plusone} and \ref{lem:downward-harm-robust} work exactly as stated in the main body of the paper (the proofs go through unchanged). 

Algorithmically, we need two facts for the truthfulness and global Pareto optimality proofs to go through (the definitions are given in Section~\ref{sec:truthfulness}. The first is that Algorithm~\ref{alg:graph} satisfies prefix-independence, which follows from the bottom-up construction of the algorithm. The specific $K$-data-levels of an agent are not used until that agent is selected by the algorithm. 

The second fact is that Algorithm~\ref{alg:graph} has the iteration-highest-response property. Fix an iteration $\tau$ and let $i$ be the agent selected at that iteration. Consider any profile in which (a) all agents selected before $\tau$ are fixed to the levels the algorithm assigned them, and (b) all agents not yet selected are weakly above $i$. The graph algorithm assigns $x_i^{\max}=\max\{\texttt{prev},\rho_i\}$ so in particular $x_i^{\max}\ge \rho_i$. By construction, $\rho_i$ is the smallest level at which $i$'s local upward incentive becomes negative in the environment where all earlier agents contribute what was already fixed and all later agents contribute weakly more. This is because it is the point where $i$ reaches its relevant $K$-data-level threshold after subtracting the already-fixed neighbors' contribution via the residual update. Therefore, since $x_i^{\max}$ is weakly above this threshold, agent $i$ strictly does not want to deviate upward from
$x_i^{\max}$ under any profile consistent with fixing the earlier agents and placing the remaining agents weakly above. So the iteration-highest-response property holds.

Now that the above set of properties is established, the proofs of Theorems~\ref{thm:graph-truthful} and~\ref{thm:graph-global-pareto} follow from the proofs of Theorems~\ref{thm:truthful} and~\ref{thm:global-pareto}.

\section{Discrete Case}
\label{app:discrete}
The discrete setting is identical to the continuous model in utilities and in the fair-exchange rule, except that each agent's data collection decision must be integral: each agent $i$ chooses $x_i \in \mathbb{Z}_{\ge 0}$.

Unlike the continuous case, not every total-access profile corresponds to an integer collection profile (though we can compute a total-access profile from any data-collection profile). Accordingly, throughout the discrete appendix we work directly in collection space $\mathbb{Z}_{\ge 0}^n$.

\subsection{Discrete Supermodularity Counterexamples}
\label{subsec:discrete-super-counter}
In the discrete model, the requirement for integer strategies breaks both supermodularity in total data space and implies that there is no Pareto-best equilibrium (equilibria may be incomparable in terms of agent utilities). We give two examples that illustrate these phenomena.

\begin{example}[Non-monotone BR in the discrete game (total data space)]
\label{ex:discrete-nonmonotone-br}
3 agents, $p_1, p_2, p_3$

Consider $p_1$ with data costs 1 and data benefit function (with $\epsilon$ small):
\[
b(T) =
\begin{cases}
(1+\epsilon)T, & T \leq 11, \\[6pt]
11(1+\epsilon), & T > 11.
\end{cases}
\]
We can interpret this as agent 1 will collect data until they have at least 11 total accessible data, then stop. 

In the first case, agents 2 and 3 collect $(0, 6)$. Then, the best response for $p_1$ is to collect $6$ and get total data $12$ (since integer data collection prevents them from getting total data $11$). 

In the second case, agents 2 and 3 collect $(1, 7)$. Then, the best response for $p_1$ is to collect $5$ and get total data $11$. So even though the opponents strategy profiles strictly increased, the best response went down (both in data collection and total data accessible space).
\end{example}

\begin{example}[No Pareto-best equilibrium]
\label{ex:discrete-incomparable-eq}
This example shows that the utilities of different equilibria may not be componentwise comparable. 

6 agents, $p_1, p_2, p_3, p_4, p_5, p_6$. 

All data costs for all agents are 1.

For agents $p_1$ and $p_2$, the benefit function is
\[
b(T) =
\begin{cases}
\left(\tfrac{1}{6}\right) T, & T \leq 6, \\[6pt]
1, & T > 6.
\end{cases}
\]

For agents $p_3, p_4$ and $p_5$, the benefit function is
\[
b(T) =
\begin{cases}
(1+\epsilon) T, & T \leq 22, \\[6pt]
(1+\epsilon)\cdot 22, & T > 22.
\end{cases}
\]

For agent $p_6$, the benefit function is
\[
b(T) =
\begin{cases}
(1+\epsilon) T, & T \leq 117, \\[6pt]
(1+\epsilon)\cdot 117 + \epsilon(T - 117), & T > 117.
\end{cases}
\]
One possible equilibrium is $X_1 = (0, 0, 6, 6, 6, 100)$ and another possible equilibrium is $X_2 = (1, 1, 5, 5, 5, 100)$. In the first equilibrium, $p_6$ is better off and in the second equilibrium $p_3, p_4, p_5$ are better off.

There is no equilibrium that dominates both of these. We know that $p_1$ and $p_2$ can collect $0$ or $1$ point at equilibrium. If they don't both collect 1, then agents $p_3, p_4, p_5$ will all collect $6$ points to hit their threshold of 22 total data each and be worse off than they are in the $X_2$. However if $p_1$ and $p_2$ do both collect 1, then the remaining players are constrained to $X_1$ and $p_6$ will be worse off. So there is no Pareto-best equilibrium in this case.
\end{example}

\subsection{Discrete Maximal Equilibrium Algorithm}
We are no longer able to take as input merely the $K$-data-levels since we need to assess best responses to specific data-collection profiles in the loop of the algorithm. 

\begin{algorithm}[H]
  \SetAlgoNoLine
  \KwIn{$c_j$ and $b_j(\cdot)$ for all $j\in[n]$}
  \KwOut{Integer contribution profile $X=(x_1,\dots,x_n)\in \mathbb{Z}^n_{\ge 0}$}

  Compute $K$-data levels $\{s_i^{\,K}\}_{K=1}^{n}$ from $c_i, b_i(\cdot)$\;
  $R \gets \{1,\dots,n\}$\;
  $\textit{prev} \gets 0$\;
  leave all $x_j$ unset\;

  \While{$R \neq \emptyset$}{
    $i \gets n-|R|$\tcp*[r]{\# agents already assigned}

    \tcp{Compute each remaining agent's ``next-stage target'' and feasible floor.}
    \ForEach{$j \in R$}{
      $\tilde x_j \gets \dfrac{s_j^{\,n-1-i} - \sum_{k\notin R} x_k}{|R|}$\;
      $m_j \gets \max\{\lfloor \tilde x_j\rfloor,\ \textit{prev}\}$\;
    }

    $m \gets \min_{j\in R} m_j$\;
    $S \gets \{j\in R : m_j = m\}$\tcp*[r]{tie group with smallest floor}

    \tcp{Try placing the entire tie group at level $m+1$.}
    \tcp{Check for a profitable \emph{downward} deviation ($m{+}1 \to m$) under a pessimistic completion where all agents in $R\setminus S$ contribute at least $m+1$.}
    \eIf{no agent in $S$ deviates downward if assigned $m+1$ under this completion}{
      \ForEach{$j \in S$}{
        $x_j \gets m+1$\;
      }
      $\textit{prev} \gets m+1$\;
      $R \gets R \setminus S$\;
    }{
      choose any $j^* \in S$ that \emph{would} deviate downward if assigned $m+1$ under this completion\;
      $x_{j^*} \gets m$\;
      $\textit{prev} \gets m$\;
      $R \gets R \setminus \{j^*\}$\;
    }
  }

  \Return{$X$}\;

  \caption{Discrete Maximal Data Equilibrium}
  \label{alg:discrete}
\end{algorithm}

\subsection{Discrete Algorithm Correctness}
\begin{theorem}[Correctness of Algorithm~\ref{alg:graph}]
\label{thm:discrete-alg-ne}
The profile $X^{opt}$ returned by Algorithm~\ref{alg:discrete} is a Nash equilibrium.
\end{theorem}
\begin{proof}
We first note that it suffices to check $\pm 1$ deviations, similar to how it suffices to check local deviations in the continuous case. For fixed $X_{-i}$, the discrete marginal gain $U_i(x_i{+}1;X_{-i})-U_i(x_i;X_{-i})$ is nonincreasing in $x_i$ due to the concavity of the benefit function and the fair exchange contract giving less data per unit collected at higher levels (since fewer agents will be above). so
$U_i(\cdot;X_{-i})$ is single-peaked on $\mathbb{Z}_{\ge0}$ and $\pm 1$ deviations suffice. 

Fix an iteration with remaining agents $R$ (so $K:=|R|$) and tie group $S\subseteq R$ (so $r:=|S|$) with candidate integer data-collection level $m$. 

Suppose we assign every agent in $S$ the common level $m{+}1$, keeping previously fixed agents unchanged. The upward deviation constraint is automatically satisfied, since these agents are weakly above their relevant $K$-data-level. The algorithm manually checks each agent for a profitable downwards deviation and only assigns the group to this level if there are none. 

Now suppose we assign some agent $i\in S$ the level $m$. The downward deviation constraint is automatically satisfied, since this agent is weakly below their relevant $K$-data-level unless \textit{prev} binds. If \textit{prev} binds, we use the state of the algorithm at that time to show no downward deviation, exactly as in the proofs of Theorem ~\ref{thm:graph-alg-ne}. We also know if all remaining agents were placed at $m+1$, $i$ would not deviate upwards, since we chose them for deviating downwards in the same environment. So in the actual profile, which is weakly less favorable since the remaining agents are only weakly above $m$, $i$ does not want to deviate upwards either. 
\end{proof}

\subsection{Discrete Truthfulness and Pareto Optimality}
\label{app:truthfulness-discrete}

\paragraph{Mechanism and reports.}
Each agent $i$ has type $v_i = (b_i,c_i)$ where $b_i:\mathbb{R}_{\ge 0}\to\mathbb{R}$ is increasing and concave and $c_i>0$ is a per-unit cost. Unlike the continuous case, it is easiest to consider the agents giving their complete cost and benefit functions. Given reports, the mechanism runs the discrete maximal-equilibrium construction (with fixed tie-breaking) and outputs an integer collection profile $x\in\mathbb{Z}_{\ge 0}^n$, which is then enforced prior to exchange.

\subsubsection{Proof Adaptation: Continuous to Discrete}
\begin{theorem}[Discrete truthfulness under enforcement]
\label{thm:discrete-truthful}
Under Model~1, truthful reporting is a Nash equilibrium of the discrete reporting game: for all $i$, all $V=(v_i,v_{-i})$, and all misreports $\hat v_i$,
\[
U_i(V) \;\ge\; U_i(\hat v_i,v_{-i}),
\]
where $U_i(\hat v_i,v_{-i})$ denotes $i$'s \emph{true} payoff when the mechanism is run on reports $(\hat v_i,v_{-i})$.
\end{theorem}
\begin{theorem}[Discrete global Pareto optimality]

\label{thm:discrete-global-pareto}
$X^{opt}$ is globally Pareto optimal. No other integer data-collection vector $X'$ achieves
\[
U_i(X') \ge U_i(X^{opt})\ \text{for all } i
\quad\text{and}\quad
U_j(X') > U_j(X^{opt})\ \text{for some}\ j.
\]
Thus the optimal fair-exchange equilibrium for the discrete case is Pareto-optimal over all strategy profiles.
\end{theorem}

Similar to the graph case, rather than reproducing full proofs, we isolate the key properties needed to adapt the continuous proofs of Theorems~\ref{thm:truthful} and~\ref{thm:global-pareto} to the discrete setting. Once these properties are established, the remainder of each argument proceeds exactly as in the continuous case.

In the discrete case, positive spillovers in data collection space are still immediate. In addition, Lemmas \ref{lem:upward-harm-fixed-lower-plusone} and \ref{lem:downward-harm-robust} are true for deviations in the new discrete strategy space. 

Algorithmically, we need two facts for the truthfulness and global Pareto optimality proofs to go through (the definitions are given in Section~\ref{sec:truthfulness}. The first is that Algorithm~\ref{alg:discrete} satisfies prefix-independence, which follows from the bottom-up construction of the algorithm. The specific $K$-data-levels of an agent are not used until that agent is selected by the algorithm. 

The second fact is that Algorithm~\ref{alg:discrete} has the iteration-highest-response property. When the algorithm assigns $x_i^{\max}=m$ to an agent $i\in S$, we have verified that $i$ does not want to increase its level by one unit in the environment where all remaining agents are weakly above them and previous agents are fixed to their levels, which is exactly the iteration-highest-response property. When the algorithm assigns a group of agents to $m+1$, they are being assigned above their $K$-data-level and thus also will not want to deviate upwards regardless of the positions of agents selected after them. Therefore the discrete algorithm satisfies iteration-highest-response.

Now that the above set of properties is established, the proofs of Theorems~\ref{thm:discrete-truthful} and~\ref{thm:discrete-global-pareto} follow from the proofs of Theorems~\ref{thm:truthful} and~\ref{thm:global-pareto}.